%% file: ms.tex
\documentclass[sigconf,natbib=false,nonacm]{acmart}
\pdfoutput=1
\renewcommand\footnotetextcopyrightpermission[1]{}
\usepackage{amssymb}

\usepackage{amsmath}

\usepackage{amsfonts}

\usepackage{amsthm}
\usepackage{bbm}
\usepackage{algorithm,algpseudocode}
\usepackage{chemarrow}
\usepackage{graphicx}
\usepackage{textcomp}
\usepackage{xcolor}
\usepackage{siunitx}
\usepackage[nolist,nohyperlinks]{acronym}

\acrodef{CSI}{channel state information}
\acrodef{CRN}{chemical reaction network}
\acrodef{GRN}{gene-regulatory network}
\acrodef{BM}{Boltzmann machine}
\acrodef{MC}{Molecular Communication}
\acrodef{IoBNT}{Internet of Bio-Nano Things}
\acrodef{ML}{maximum-likelihood}
\acrodef{MAP}{maximum-a-posteriori}
\acrodef{ISI}{inter-symbol interference}
\acrodef{BER}{bit error rate}
\acrodef{CTMC}{continuous-time Markov Chain}
\acrodef{SNR}{signal-to-noise ratio}
\acrodef{RV}{random variable}
\acrodef{BCSK}{binary concentration shift keying}
\acrodef{FVBM}{fully visible Boltzmann machine}
\acrodef{iid}[i.i.d.]{independent and identically distributed}
\acrodef{SCRN}{stochastic chemical reaction network}
\acrodef{iff}{if and only if}

\newcommand{\xhat}{\hat{x}}

\newcommand{\Xhat}{\hat{X}}
\newcommand{\y}{\mathbf{y}}
\newcommand{\Y}{\mathbf{Y}}
\newcommand{\z}{\mathbf{z}}
\newcommand{\Z}{\mathbf{Z}}
\newcommand{\partition}{\mathcal{Z}}

\newcommand{\expectation}[2]{\mathbb{E}_{#1}\{#2\}}

\newcommand{\Xhaton}{\hat{X}^{\text{ON}}}
\newcommand{\Xhatoff}{\hat{X}^{\text{OFF}}}
\newcommand{\XtrueOn}{X^{\text{ON}}_{\text{pilot}}}
\newcommand{\XtrueOff}{X^{\text{OFF}}_{\text{pilot}}}

\newcommand{\Yon}[1]{Y^{\text{ON}}_{#1}}
\newcommand{\Yoff}[1]{Y^{\text{OFF}}_{#1}}
\newcommand{\Won}{W^{\text{ON}}}
\newcommand{\Woff}{W^{\text{OFF}}}

\newcommand{\threshold}{\nu}

\newcommand{\Nr}{N_{\text{r}}}
\newcommand{\Nrb}{N_{\text{r,b}}}
\newcommand{\Nw}{N_{\text{w}}}
\newcommand{\Nwa}{N_{\text{w}}^{\text{ON}}}
\newcommand{\Nwi}{N_{\text{w}}^{\text{OFF}}}
\newcommand{\wxy}{w_{\text{x,y}}}

\newcommand{\cx}{c_x}
\newcommand{\cn}{c^{\text{n}}}
\newcommand{\deltac}{\Delta_{\text{c}}}

\newcommand{\rrc}[1]{k_{\text{#1}}} %

\newcommand{\xtrue}{ x_{ \text{pilot} } }
\newcommand{\meter}{\text{m}}
\newcommand{\second}{\text{s}}
\newcommand{\crn}[1]{\mathcal{C}_{\mathrm{#1}}}

\newcommand{\transpose}{\intercal}
\newcommand{\W}{\mathbf{W}}
\newcommand{\V}{\mathbf{V}}

\newtheorem{theorem}{Theorem}

\AtBeginDocument{%
  }

\setcopyright{acmcopyright}
\copyrightyear{2023}
\acmYear{2023}
\acmDOI{XXXXXXX.XXXXXXX}
\acmConference[NanoCom '23]{10th ACM International Conference on Nanoscale Computing and Communication}{September 20--22,
  2023}{Coventry, UK}
\acmPrice{15.00}
\acmISBN{978-1-4503-XXXX-X/18/06}

\RequirePackage[
  datamodel=acmdatamodel,
  style=acmnumeric,
  sorting=none
  ]{biblatex}

\begin{document}

\title{Stochastic Chemical Reaction Networks for MAP Detection in Cellular Receivers}

\author{Bastian Heinlein\textsuperscript{*}, Lukas Brand\textsuperscript{*}, Malcolm Egan\textsuperscript{\dag}, Maximilian Schäfer\textsuperscript{*}, Robert Schober\textsuperscript{*}, and Sebastian Lotter\textsuperscript{*}}
\affiliation{%
    \institution{\textsuperscript{*}Friedrich-Alexander-Universität Erlangen-Nürnberg, Erlangen, Germany}\country{}
}
\affiliation{%
    \institution{\textsuperscript{\dag}Univ. Lyon, Inria, INSA Lyon, Villeurbanne, France}\country{}
}
\email{{bastian.heinlein,lukas.brand,max.schaefer,robert.schober,sebastian.g.lotter}@fau.de,malcom.egan@inria.fr}

\renewcommand{\shortauthors}{Heinlein et al.}

\begin{abstract}
In order to fully exploit the potential of molecular communication (MC) for intra-body communication, practically implementable cellular receivers are an important long-term goal. A variety of receiver architectures based on chemical reaction networks (CRNs) and gene-regulatory networks (GRNs) has been introduced in the literature, because cells use these concepts to perform computations in nature. However, practical feasibility is still limited by stochastic fluctuations of chemical reactions and long computation times in GRNs. Therefore, in this paper, we propose two receiver designs based on stochastic CRNs, i.e., CRNs that perform computations by exploiting the intrinsic fluctuations of chemical reactions with very low molecule counts. The first CRN builds on a recent result from chemistry that showed how Boltzmann machines (BMs), a commonly used machine learning model, can be implemented with CRNs. We show that BMs with optimal parameter values and their CRN implementations can act as maximum-a-posteriori (MAP) detectors. Furthermore, we show that BMs can be efficiently trained from simulation data to achieve close-to-MAP performance. While this approach yields a fixed CRN once deployed, our second approach based on a manually designed CRN can be trained with pilot symbols even within the cell and thus adapt to changing channel conditions. We extend the literature by showing that practical robust detectors can achieve close-to-MAP performance even without explicit channel knowledge.
\end{abstract}

\maketitle

\pagestyle{plain}

\input{sections/introduction}
\input{sections/BM_CRN_Introduction}
\input{sections/system_model}
\input{sections/MAP_BM}
\input{sections/MAP_CRN}
\input{sections/performance_evaluation}
\input{sections/conclusion}

\printbibliography

\end{document}

%% file: sections/introduction.tex
\section{Introduction}
\label{sec:introduction}
\ac{MC} is a new paradigm for information exchange in conditions that are unfavorable for traditional wireless communication, e.g., at nano-scale or inside the human body. \ac{MC} has great potential, for example in the context of the Internet of Bio-Nano Things which will enable groundbreaking improvements in healthcare by employing a network of connected nano-sensors within the body to diagnose and treat diseases \cite{haselmayr2019_MC_6G}.

For both theoretical work \cite{eckford_survey} and experimental testbeds \cite{lotter_testbeds} progress has been made to realize this ambitious vision. A further leap forward could be possible if practical implementations of receivers for cell-to-cell communication were found.

In contrast to electromagnetic wave-based based mobile communication, there are no general-purpose processors available to implement receivers for \ac{MC} on the scale of individual cells. Yet, in nature, cells do communicate with each other using signaling molecules. Often, cells have ligand-binding receptors on the surface and use \acp{CRN}, i.e., networks of interacting chemical species, to perform computations. For example, bacteria use \acp{CRN} to determine whether to tumble or to move forward in the context of chemotaxis. Also, receptor states can cause long-term changes by interacting with the cells' \acp{GRN}, i.e., by influencing which genes are expressed depending on the sensed environmental conditions.

Therefore, it is a natural choice to formulate signal detection problems in such a way that \acp{CRN} and \acp{GRN} can be used to perform the computations required for artificial \ac{MC}. 

In \cite{kuscu2018_ml_detection,chou2019designing_molecular_circuits,egan_2018_detector}, \acp{CRN} were designed manually for a given channel to implement receiver components. However, these approaches either do not account for the stochastic fluctuations of \acp{CRN}, require many molecules, or are suboptimal compared to \ac{MAP} detectors.

The authors of \cite{anderson2021reaction} used \acp{CRN} to implement feed-forward neural networks. This idea is especially interesting for \ac{MC} if analytical channel models are not available or their parameters are unknown. However, the resulting \acp{CRN} are very complex and the impact of stochastic fluctuations has not been considered in \cite{anderson2021reaction}.

The use of \acp{GRN} to perform computations in general has been reviewed in \cite{saltepe_biosensors} and proposed for receiver design in \cite{unluturk2015genetically_engineered_bacteria}. While this approach is especially interesting because it can exploit natural mechanisms already present in cells and it has been shown that logical functions can be implemented via \acp{GRN}, very long time scales are needed to perform complex computations. A recent review highlighted that using transcriptional elements is only feasible for very simple receiver implementations in the near future \cite{femminella2022_implementation_issues}.

Finally, several concepts for implementing molecular machine learning based on \acp{CRN}, bacterial multi-species communication, and Calcium signaling were proposed in \cite{balasubramaniam2022_molecule_ml}. Potentially, these approaches could be also applied for receiver design. However, it is unclear how complex the resulting architectures would be and how long the computations would take.

A common issue for all previously mentioned approaches is the diversity of transmission channels. Each channel involves different parameters and possibly time-variant parameter values. Yet, most model-based approaches for receiver design in the literature are not concerned with estimating correct parameters and adapting to changing channel conditions. While the existing learning-based approaches could in principle remedy this issue, it remains unclear whether their training would be fast enough to adapt to changing channel conditions.

In this paper, we propose the practical implementation of \ac{MAP} detectors for cellular receivers with ligand-binding receptors using two different \acp{CRN}. While the stochastic fluctuations of molecule counts are usually considered as noise source in \acp{CRN}, our designs exploit this randomness to perform computations. To emphasize this, we call the proposed \acp{CRN} ''stochastic''.

For the first \ac{CRN}, we start by approximating the joint distribution of the transmitted symbol and the receptor states using a \ac{BM}. \acp{BM} are commonly used graphical models in machine learning to approximate arbitrary joint probability distributions of binary \acp{RV} \cite{mackay2003information}. Then, we exploit the methods reported in \cite{poole2017_chemical_BMs} to represent \acp{BM} via stochastic \acp{CRN}. In order to obtain posterior estimates for the transmitted symbol, i.e., to perform detection, we then condition the \ac{BM}, or rather its \ac{CRN} representation, on the receptor states.

The second \ac{CRN} proposed in this paper is not derived from \acp{BM}, but is designed directly for our purposes. It has a lower complexity and can adapt to changing channel conditions over time by learning efficiently from pilot symbols. We demonstrate the performance of the proposed approach for an \ac{MC} channel subject to time-varying background noise levels.

The remainder of this paper is organized as follows. In Section \ref{sec:BM_CRN_introduction}, we review the fundamentals of \acp{BM}, \acp{CRN}, and how \acp{BM} can be represented via \acp{CRN}. In Section \ref{sec:system_model}, we introduce the considered system model and the corresponding \ac{MAP} detector. In Sections \ref{sec:BM_MAP} and \ref{sec:low_complexity_MAP}, we introduce the proposed receiver designs which are then evaluated and compared to the optimal \ac{MAP} detector in Section \ref{sec:performance_evaluation}. Finally, we summarize our main findings and potential future work in Section \ref{sec:conclusion}.

%% file: sections/BM_CRN_Introduction.tex
\section{Boltzmann Machines and Chemical Reaction Networks}
\label{sec:BM_CRN_introduction}
\subsection{Boltzmann Machines}
In order to approximate a joint distribution of $N$ binary \acp{RV} with probability mass function $q_{\Z}(\z)$, a \ac{BM} with at least $M \geq N$ nodes is required. $N$ of these nodes are each identified with a binary \ac{RV} of the original distribution. The additional $M-N$ nodes are ''hidden'', i.e., they do not correspond to a specific \ac{RV} but simply allow for additional degrees of freedom in the \ac{BM} in order to achieve a higher approximation quality. However, as we will show in Section \ref{sec:BM_MAP}, no hidden nodes are required for the \ac{MAP} detection task in this paper. Therefore, we restrict ourselves to the simpler case of a \ac{BM} without hidden nodes, which we call \ac{FVBM}.

An \ac{FVBM} consists of $N$ nodes, each of which corresponds to a binary \ac{RV} $Z_i,i\in\{1,\dots,N\}$, which we collect in a vector $\Z=[Z_1,\ldots,Z_N]^\transpose$, where $[\cdot]^\transpose$ denotes the transpose operator. The probability mass function of the \ac{FVBM} nodes is then given by
\begin{equation}
    p_{\Z}(\z) = \frac{1}{\partition} \exp\left(\frac{1}{2} \z^\transpose \W \z + \z^\transpose \theta \right).
\end{equation}
Here, $\theta \in \mathbb{R}^{N \times 1}$ and $\W \in \mathbb{R}^{N \times N}$ denote respectively the vector of biases and a symmetric weight matrix with all-zero diagonal, which captures the correlations between the nodes. $\partition$ is a normalization constant that ensures that $p_{\Z}(\z)$ is a probability distribution.

The state of node $i$ depends on its associated bias $\theta_i$ and the current state of all other nodes $j \neq i$ with non-zero correlations to node $i$, captured by weight matrix entries $W_{i,j}$. Formally, the probability that $Z_i=1$ is given by 
\begin{equation}
    p_{Z_i|\Z_{-i}}(Z_i=1|\Z_{-i}=\z_{-i}) = \sigma\left(\theta_i + \sum_{j \neq i} W_{i,j} z_j\right) \label{eq:bm_conditional_probability},
\end{equation}
where $\Z_{-i}$ is the vector of all \acp{RV} except the one associated with node $i$, $\z_{-i}$ contains their observed realization, and $\sigma(x)=\frac{1}{1+e^{-x}}$. Eq. (\ref{eq:bm_conditional_probability}) is for example well known from the Gibbs sampling algorithm \cite{mackay2003information}.

When using a \ac{BM} to approximate $q_{\Z}(\z)$, $\W$ and $\theta$ can be learned from the first- and second-order moments $\expectation{q}{\z}$ and $\expectation{q}{\z\z^\transpose}$, respectively, where $\expectation{q}{\cdot}$ denotes the expectation operator for a probability distribution $q_{\Z}(\z)$. For an introduction to learning weights and biases for \acp{BM}, we refer to \cite{mackay2003information}.

\subsection{Chemical Reaction Networks}
As previously mentioned, \acp{BM} can be realized as \acp{CRN} \cite{poole2017_chemical_BMs}. Formally, a \ac{CRN} $\crn{}=(\mathcal{S},\mathcal{R},\mathbf{k})$ consists of a set of species $\mathcal{S}$, a set of reactions $\mathcal{R}$ defined over $\mathcal{S}$, and a vector of reaction rate constants $\mathbf{k}$.
\\
A chemical reaction $r \in \mathcal{R}$ converts molecules into other molecules. For mass-action kinetics, the propensity of reaction $r$, i.e., how likely it occurs per unit time, is proportional to the reaction rate constant $k_r>0$ and the number of available reactants in the considered reaction network. 

Consider for example the reaction
\begin{equation}
    \Won + \Xhaton \xrightarrow{k_r} \Won + \Xhatoff\label{eq:example_reaction}.
\end{equation}
Here, the propensity of reaction $r$ is given by the product of the reaction rate constant, the number of $\Won$ molecules, and the number of $\Xhaton$ molecules. Thus, if there is no $\Xhaton$ or no $\Won$ molecule, the reaction cannot happen at all.

In the remainder of this paper, we often write that a molecule can be in one of two states, namely ON or OFF. This is to represent binary states, e.g., of a \ac{RV}, by molecules. Chemically speaking, the ON-version of a molecule might for example contain a phosphor group that is not contained in the complementary OFF species. Otherwise, the molecules are identical. We also say that the ON species is ''active'' whereas the OFF species is ''inactive''. 

Thus, we say for example that the $\Won$ molecule deactivates the $\Xhaton$ molecule in reaction $r$ in (\ref{eq:example_reaction}).
To implement a \ac{BM} using a \ac{CRN}, we exploit the fact that a \ac{CRN} can be described by a \ac{CTMC} \cite{poole2017_chemical_BMs}. We assume a \ac{CRN} with species set $\{Z_1^{\mathrm{ON}},Z_1^{\mathrm{OFF}},\dots,Z_N^{\mathrm{ON}},Z_N^{\mathrm{OFF}}\}$ and demand that there is exactly either one $Z_i^{\mathrm{ON}}$ or one $Z_i^{\mathrm{OFF}}$ molecule for $\forall i$ at any point in time. Then, we can say that the molecule $Z_i$ is ON if the $i$-th node of the corresponding \ac{BM} has the value $1$ whereas the $Z_i$ molecule is OFF if the node has value $0$ at a given point in time. By defining appropriate reactions that activate or deactivate molecules, one can ensure that the \ac{CTMC} describing the \ac{CRN} eventually reaches a stationary distribution and that this stationary distribution is the same as the one of the corresponding \ac{BM}. In this case, the \ac{CRN} is said to implement the \ac{BM}.

This idea has been introduced and formalised in \cite{poole2017_chemical_BMs}. In fact, in \cite{poole2017_chemical_BMs} three representations of \acp{BM} using \acp{CRN} are provided. Two of them are exact but require a number of reactions $|\mathcal{R}|$ that scales exponentially with the number of nodes of the \ac{BM}, where $|\mathcal{X}|$ denotes the cardinality of a set $\mathcal{X}$. The third one is an approximation of the \ac{BM}, which requires much fewer reactions at the cost of a mismatch between the stationary distributions of the \ac{CRN} and the \ac{BM}.

%% file: sections/system_model.tex
\section{System Model}
\label{sec:system_model}
\begin{figure}
    \centering
    \includegraphics[width=0.5\textwidth]{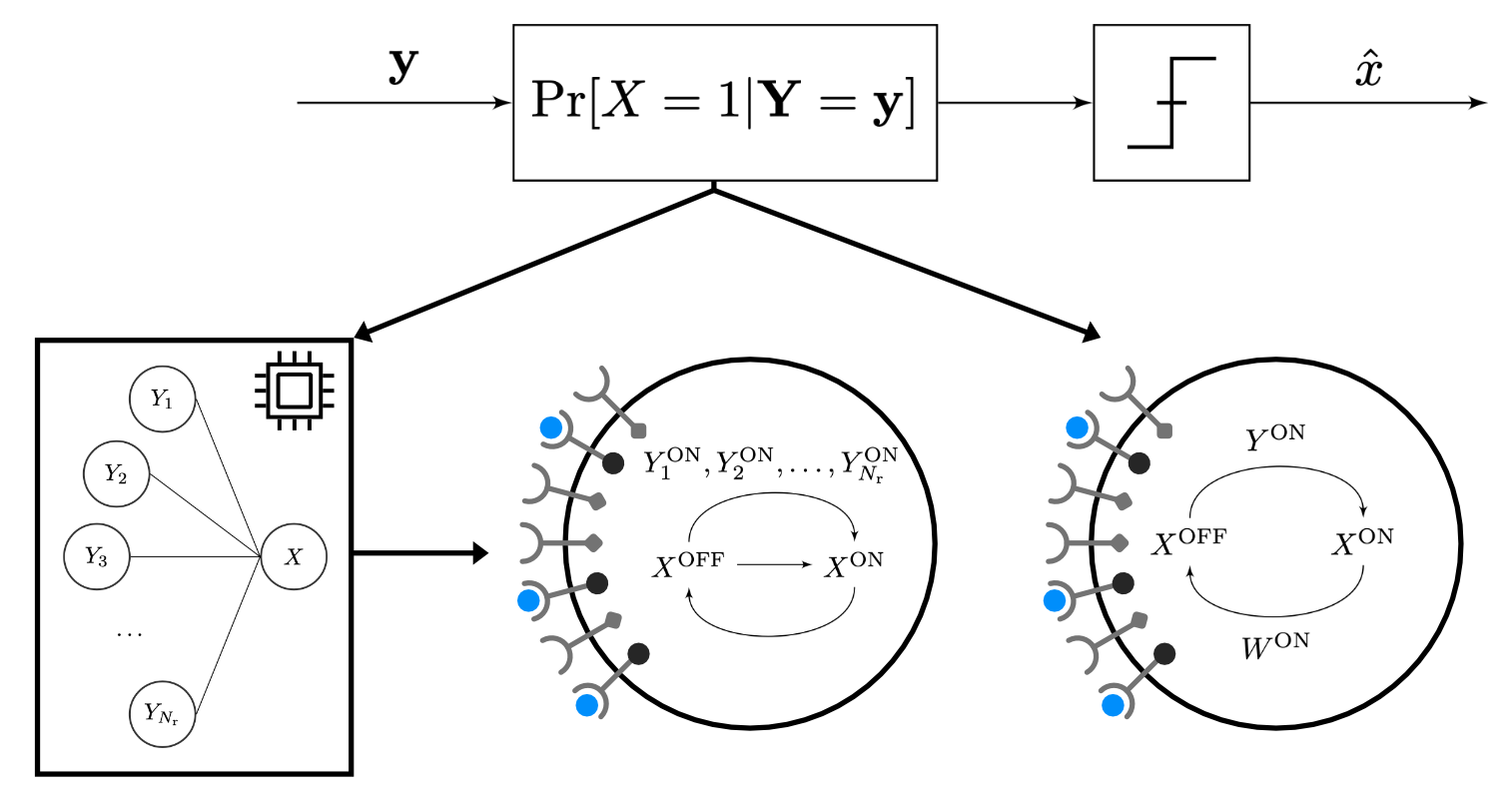}
    \vspace*{-7mm}
    \caption{Proposed receiver model. The probabilistic model can be implemented in silico using a \ac{BM} (bottom left). Its in vivo implementation could be a \ac{CRN} based on the \ac{BM} (bottom center) or a directly designed \ac{CRN} (bottom right).}
    \label{fig:system_model}
    \vspace*{-5mm}
\end{figure}
As our focus is the detection process, we keep the modulation scheme and channel model simple. Specifically, we assume the transmission of a binary source symbol $x \in \{0,1\}$ via \ac{BCSK} and the absence of \ac{ISI}. The latter holds for sufficiently long symbol intervals or enzymatic ligand degradation \cite{noel2014_enzymes}. Considering that existing receiver implementations may require extremely long decoding times of up to several hours per symbol \cite{unluturk2015genetically_engineered_bacteria,femminella2022_implementation_issues}, choosing long intervals to avoid \ac{ISI} would not be a performance bottleneck compared to existing approaches.

A cellular receiver, as depicted in Fig. \ref{fig:system_model}, senses its environment through $\Nr$ cell surface receptors. Depending on whether receptor $i$ is bound to a ligand, its intra-cellular domain might have different chemical properties. Thus, we represent receptor $i$'s intra-cellular domain by different chemical species depending on its state, namely by $\Yon{i}$ if the receptor is bound to a ligand and by $\Yoff{i}$ otherwise. For our purposes, we assume that the receptor states are sampled once and then stored in a vector $\y \in \{0,1\}^{\Nr}$ where $y_i=1$ if we observe $\Yon{i}$ and $y_i=0$ for $\Yoff{i}$. Mathematically, the considered \ac{MC} channel is characterized by the joint distribution $q_{\Y,X}(\y,x)$ where $X$ denotes the binary \ac{RV} corresponding to the transmitted symbol and $\Y$ contains the binary \acp{RV} $Y_i$ corresponding to the different receptors. 

A \ac{MAP} detector for \ac{BCSK} would compute the estimated symbol $\xhat$ as
\begin{equation}
    \xhat = \xhat^{\mathrm{MAP}} = \begin{cases}
        1 & ,\;\;\text{if} \Pr[X=1|\Y=\y]\geq\frac{1}{2}\\
        0 & ,\;\;\text{otherwise}
    \end{cases}.\label{eq:xhat_probabilistic_model}
\end{equation}
When the detection \eqref{eq:xhat_probabilistic_model} is performed by a deterministic device, $\xhat^{\mathrm{MAP}}$ can be obtained directly.
However, deterministic devices may be difficult to realize on the single cell level.
Hence, in this paper we assume that, rather than observing $\xhat^{\mathrm{MAP}}$ directly, independent samples of a generative\footnote{We use the term ''generative model'' to refer to the generation of samples from a conditional probability distribution $\Pr[\hat{X}=\xhat|\Y=\y]$, which is not to be confused with the different usage of the same term in the context of statistical classification.}, binary probability model, i.e., a binary \ac{RV} $\Xhat$ with probability mass function $f_{\Xhat}(\y)$, are observed and $\xhat$ is obtained by averaging these samples.
In the special case that $f_{\Xhat}(\y)=\Pr[X=1|\Y=\y]$, according to the law of large numbers $\xhat = \xhat^{\mathrm{MAP}}$ when sufficiently many samples are drawn; yet, we will see that there exist $f_{\Xhat}(\y) \neq \Pr[X=1|\Y=\y]$ for which $\xhat=\xhat^{\mathrm{MAP}}$, i.e., the choice of the probability model for $\Xhat$ that induces \ac{MAP} detection is not unique.

In the remainder of this paper, we assume that all receptors have identical chemical properties, such that $\Pr[X=1|\Y=\y]$ depends only on the number of bound receptors $\Nrb=\sum_{i=1}^{\Nr} y_i$. 
We further assume that the likelihoods $\Pr[\Nrb=n|X=x]$ have exactly one local maximum $n_{\mathrm{m},x}$ and are monotonically increasing for $n \leq n_{\mathrm{m},x}$ and monotonically decreasing for $n > n_{\mathrm{m},x}$.

For these quite general assumptions, the \ac{ML} detector is a simple threshold detector, i.e.,
\begin{equation}
    \xhat^{\mathrm{ML}} = \begin{cases}
        1 & \text{,\;\;\;if } \Nrb \geq \threshold \\
        0 & \text{,\;\;\;otherwise}
    \end{cases},
    \label{eq:map_detector}
\end{equation}
for some threshold $\threshold \in \mathbb{N}_0$. Moreover, we assume equiprobable symbols $x$, such that the \ac{ML} and the \ac{MAP} detectors coincide.

As already alluded to, it is not necessary to perfectly know $\Pr[X=x|\Y=\y]$ in order to make \ac{MAP} decisions. Instead, it is sufficient to find a {\em surrogate} probability mass function $f_{\Xhat}(n) \in [0,1]$ that fulfills $f_{\Xhat}(n)\geq\frac{1}{2}$ if and and only if $\Pr[X=1|\Nrb=n]\geq\frac{1}{2}$.
Such $f_{\Xhat}(n)$, when employed as surrogate model of the true posterior distribution, allows for \ac{MAP} detection.
We say that such a probability model has the \textit{MAP property}.
The \ac{MAP} property will be useful in the following when different generative probability models are analyzed with respect to their potential as \ac{MAP} detectors.

%% file: sections/MAP_BM.tex
\section{Boltzmann Machine-Inspired MAP Detectors}
\label{sec:BM_MAP}
In this section, we explore \ac{BM}-based \ac{MAP} detection.
To this end, we first verify that \acp{BM} can be used as generative models to implement \ac{MAP} detection.
Then, we extend this idea to \ac{BM}-inspired \acp{CRN}.

\subsection{MAP-Capability of Boltzmann Machines}
In order to use a \ac{BM} for \ac{BCSK} detection, we set $\Z=[\hat{X} \;\;\; \Y]^\transpose$.

\begin{theorem}\label{thm:bm_map}
For a known threshold $\nu$ and properly chosen parameters, \acp{BM} have the \ac{MAP} property. \label{th:bm_map_capability}
\end{theorem}
\begin{proof}
Let $Z_i=\hat{X}$, $\Z_{-i}=\Y$ in (\ref{eq:bm_conditional_probability}) and  $W_{i,j} = \wxy$ for all $i,j$, $i\neq j$, where $\wxy$ is any constant independent of $X$ and $\Y$. Then,
\begin{align}
    p_{\hat{X}|\Y}(\hat{X}=1|\Y=\y) &= \sigma\left( \theta_i + \sum_{i=1}^{N_r} \wxy y_i \right) \\
    &= \sigma\left(\theta_i + \Nrb \wxy\right).
\end{align}
Setting $\theta_i=-(\threshold-\frac{1}{2})\wxy$, the \ac{MAP} property follows, since $p_{\hat{X}|\Y}(\hat{X}=1|\Y=\y)>0.5$ if $\Nrb\geq\threshold$ and $p_{\hat{X}|\Y}(\hat{X}=1|\Y=\y)<0.5$ if $\Nrb < \threshold$.
\end{proof}
Theorem~\ref{thm:bm_map} verifies that MAP detection can be realized by a \ac{BM} if $\nu$ is known.

\subsection{Representation via Chemical Reaction Networks}
Among the mappings of \acp{BM} to \acp{CRN} proposed in \cite{poole2017_chemical_BMs}, the \textit{Edge Species Mapping} is {\em exact} and the corresponding \ac{CRN} is {\em trainable}.
However, applied to the \ac{BM} defined in the previous section, it incurs $(\Nr+1) \cdot 2^{\Nr+1}$ chemical reactions, each with up to $2\Nr+1$ reactants, which, even for a few dozen receptors, would be infeasible to realize in practice. 
Hence, we resort to the \textit{Taylor Mapping} \cite{poole2017_chemical_BMs} instead.
This mapping is {\em approximate} and, to the best of our knowledge, {\em no training} algorithm for the resulting \ac{CRN} $\crn{TM} = (\mathcal{S}_{\mathrm{TM}},\mathcal{R}_{\mathrm{TM}},\mathbf{k}_{\mathrm{TM}})$ has been proposed, yet.
However, $\crn{TM}$ incurs only $|\mathcal{R}_{\mathrm{TM}}|=2 \Nr^2 + 4 \Nr + 2$ reactions, each with at most two reactants.
Hence, it can be realized more easily than the {\em Edge Species Mapping}-based \ac{CRN}.
Furthermore, to perform inference, only a subset of $\mathcal{R}_{\mathrm{TM}}$ is required; namely only those reactions $\mathcal{R}_{\mathrm{TM},\hat{X}} \subseteq \mathcal{R}_{\mathrm{TM}}$ that activate or deactivate the molecule identified with the $\hat{X}$-node of the \ac{BM}, i.e., $Z_1$.

$\mathcal{R}_{\mathrm{TM},\hat{X}}$ encompasses only the following reactions
\begin{equation}
    \begin{split}
        \Xhatoff &\autorightleftharpoons{\small{$k$}}{\small{$k(1+|\theta_x|)$}} \Xhaton\\
        \Yon{1} + \Xhatoff &\xrightarrow{k W_{y_1,x}} \Yon{1} + \Xhaton\\
        &\dots\\
        \Yon{N_r} + \Xhatoff &\xrightarrow{k W_{y_{N_r},x}} \Yon{N_r} + \Xhaton\label{eq:BM_CRN},
    \end{split}
\end{equation}
and we denote by $\mathcal{C}_{\mathrm{TM},\hat{X}} = (\mathcal{S}_{\mathrm{TM},\hat{X}},\mathcal{R}_{\mathrm{TM},\hat{X}},\mathbf{k}_{\mathrm{TM},\hat{X}})$ the \ac{CRN} that results from restricting $\crn{TM}$ to the set of species, reactions, and reaction rate constants involved in \eqref{eq:BM_CRN}.
In \eqref{eq:BM_CRN}, $\theta_x$ is equal to the bias associated with the $\hat{X}$-node, $W_{y_i,x}$ is equal to the weight between the $Y_i$- and the $\hat{X}$-node of the BM, i.e., $W_{i+1,1}$ in \eqref{eq:bm_conditional_probability}, and $k$ is an arbitrary positive scaling factor for the reaction rate constants. $\Xhaton$ and $\Xhatoff$ are the chemical species associated with node $\hat{X}$. At any point in time, there is either a single $\Xhaton$ molecule or a single $\Xhatoff$ molecule. When no receptors are activated, i.e., $Y_i=0$ for all $i$, no $\Yon{i}$ molecules exist, and $\Xhaton$ switches to $\Xhatoff$ more likely than vice versa (assuming $\theta_x \neq 0$). On the other hand, the more receptors are activated, the more likely an $\Xhatoff$ molecule is converted to an $\Xhaton$ molecule by one of the reactions in \eqref{eq:BM_CRN}.

\begin{theorem}
    \label{th:bm_map_mc}
    $\crn{TM,X}$ preserves the MAP property of the BM.
\end{theorem}
\begin{proof}
    Let $\theta_x=-(\threshold-\frac{1}{2})\wxy$ and $W_{y_i,x}=\wxy$ as in the proof of Theorem~\ref{th:bm_map_capability} and assume $\Nrb$ is constant.
    To proof the \ac{MAP} property, we need to compute the steady-state probability to observe an $\Xhaton$ molecule. To this end, we first sum up all the rate constants for all possible reactions converting $\Xhatoff$ to $\Xhaton$ in (\ref{eq:BM_CRN}). This yields $k(1+\Nrb \wxy)$. On the other hand, $\Xhaton$ is converted to $\Xhatoff$ with rate constant $k(1+|\theta_x|)$.
    Hence, the following detailed balance condition holds in steady-state
    \begin{equation}
        \Pr[\Xhaton|\Nrb] \cdot k(1+|\theta_x|) = \Pr[\Xhatoff|\Nrb] \cdot k(1+\Nr \wxy),\label{eq:proof_thm2}
    \end{equation}
    where $\Pr[\Xhaton|\Nrb]$ denotes the probability to observe $\Xhaton$ at any given point in time if $\Nrb$ receptors are bound. By substituting $\Pr[\Xhatoff|\Nrb]=1-\Pr[\Xhaton|\Nrb]$ in \eqref{eq:proof_thm2}, we obtain
    \begin{equation}
        \Pr[\Xhaton|N_{r,b}] = \frac{1+\Nrb \wxy}{2+\Nrb \wxy+|\theta_x|}.
    \end{equation}
    Solving $\Pr[\Xhaton|\Nrb]\geq\frac{1}{2}$ for $\Nrb \wxy$ yields
    \begin{equation}
        \Nrb\wxy \geq |\theta_x|.
    \end{equation}
    Since $\theta_x=-(\threshold-\frac{1}{2})\wxy$, $\Pr[\Xhaton|\Nrb]\geq0.5$ holds \ac{iff} $\Nrb \geq\threshold$ receptors are bound.
\end{proof}

Theorem~\ref{th:bm_map_mc} confirms that \ac{MAP} detection can be performed based on observations of $\Xhaton$ generated according to \eqref{eq:BM_CRN}.
Specifically, if the detected symbol $\xhat_M$ is obtained from $M$ samples of $\Xhaton$, $[\Xhaton_1,\ldots,\Xhaton_M] \in \lbrace0,1\rbrace^M$, as
\begin{equation}
    \xhat_M = \begin{cases}
      1\quad,\quad\textrm{if}\, \frac{1}{M}\sum_{i=1}^{M} \Xhaton_i \geq \frac{1}{2}\\
      0\quad,\quad\textrm{otherwise},
    \end{cases}\label{eq:det_from_samples}
\end{equation}
then $\lim\limits_{M \to \infty} \xhat_M = \xhat^{\textrm{MAP}}$.

The \ac{MAP} detector proposed in this section does not require an analytical channel model because the \acp{BM} can be trained from data, that can be obtained, e.g., from simulations, alone. However, once deployed, its \ac{CRN} implementation cannot adapt to changes of $\threshold$ anymore. To remedy this, we propose an adaptive \ac{CRN}-based \ac{MAP} detector in the next section.

%% file: sections/MAP_CRN.tex
\section{Trainable Low-Complexity MAP Detectors}
\label{sec:low_complexity_MAP}
\subsection{A Custom MAP-Capable CRN}
For the \ac{CRN} $\mathcal{C}_{\mathrm{TM},\hat{X}}$ introduced in the previous section, the number of reactions scales linearly with the number of receptors $N_r$.
This complexity can be further reduced by exploiting that some of the reaction rate constants in \eqref{eq:BM_CRN} can be identical without sacrificing the \ac{MAP} property of \eqref{eq:BM_CRN}, see also the proof of Theorem~\ref{th:bm_map_mc}. Consequently, in this section we introduce a \ac{CRN} $\crn{LC}$ that requires only a constant number of reactions. Furthermore, we also propose a learning algorithm that allows for online adaptation of the \ac{CRN}, i.e., $\threshold$ does not have to be known at design time, and $\threshold$ may even change during operation of the \ac{CRN}.

In the previous section, we used different chemical species $\Yon{i}$ and $\Yoff{i}$ for different receptors. However, this is actually not necessary because we assume that all receptors are identical in our system and the number of bound receptors $\Nrb$ contains all information required for \ac{MAP} detection (cf. Section~\ref{sec:system_model}). Thus, we assume from now on that all receptors are of the same species ($\Yon{}$ or $\Yoff{}$ depending on their state) and drop the index. 

Let us define the following chemical reaction in which $\Yon{}$ activates $\Xhatoff$ with reaction rate constant $\rrc{on}$
\begin{equation}
    \Yon{} + \Xhatoff \xrightarrow{\rrc{on}} \Yon{} + \Xhaton.\label{eq:MAP_CRN1}
\end{equation}

In addition, we introduce another chemical species, $\Won$, that deactivates $\Xhaton$ with reaction rate constant $\rrc{off}$ as follows
\begin{equation}
    \Won   + \Xhaton  \xrightarrow{\rrc{off}} \Won   + \Xhatoff \label{eq:MAP_CRN2}.
\end{equation}
Collecting reactions \eqref{eq:MAP_CRN1} and \eqref{eq:MAP_CRN2} in $\mathcal{R}_{\textrm{LC}}$, $\crn{LC}$ is defined as $\crn{LC} = (\lbrace\Yon{},\Xhatoff, \Xhaton\rbrace, \mathcal{R}_{\textrm{LC}}, \lbrace\rrc{on},\rrc{off}\rbrace)$.
The following theorem shows that $\crn{LC}$ is indeed a sensible definition.
\begin{theorem}\label{thm:custom_crn}
    Let $\Nwa$ denote the number of $\Won$ molecules. Then, for an appropriate value of $\Nwa$, $\crn{LC}$ as defined above has the MAP property. 
\end{theorem}
\begin{proof}
    We take the same detailed balance-based approach to compute the steady-state probability of observing $\Xhaton$ as in the proof of Theorem~\ref{th:bm_map_mc}.
    Here, the rate with which $\Xhatoff$ is converted to $\Xhaton$ is given by $\Nrb \cdot \rrc{on}$.
    $\Xhaton$ is converted to $\Xhatoff$ with rate $\Nwa \cdot \rrc{off}$. 
    Hence, we obtain
    \begin{equation}
        \Pr[\Xhaton|\Nrb,\Nwa] = \frac{\Nrb}{\Nrb+\frac{\rrc{off}}{\rrc{on}}\Nwa}.\label{eq:MAP_CRN_Pr}
    \end{equation}
    Solving $\Pr[\Xhaton|\Nrb,\Nwa] \geq 0.5$ for $\Nrb$ yields $\Nrb \geq \frac{\rrc{off}}{\rrc{on}}\Nwa$.
    For example, for $\frac{\rrc{off}}{\rrc{on}}=1$, the conditions for the \ac{MAP} property are fulfilled if $\Nwa=\threshold$.
\end{proof}
Based on observations of $\Xhaton$, $\xhat$ is obtained as in \eqref{eq:det_from_samples}.
As we have seen in the proof of Theorem~\ref{thm:custom_crn}, the number of $\Won$ molecules, $\Nwa$, determines the steady-state distribution of $\crn{LC}$.
In the following section, we will exploit this to design a learning rule for $\crn{LC}$.

\subsection{Pilot Symbol-based Learning Rule}
In order to learn the optimal value of $\Nwa$ online, pilot symbols $\xtrue \in \lbrace 0,1 \rbrace$ are transmitted from the transmitter to the receiver in addition to the data symbols.
The pilot symbols are represented by the chemical species $\XtrueOff$ and $\XtrueOn$, i.e., one $\XtrueOff$ molecule is present in the receiver \ac{iff} a $0$ has been transmitted and one $\XtrueOn$ molecule \ac{iff} a $1$ has been transmitted.
Based on the pilot symbols, a learning rule for $\Nwa$ is proposed as follows.
For each $\xtrue$, if $\xhat \neq \xtrue$, $\Nwa$ is adapted towards the value of $\Nwa$ required to achieve \ac{MAP} performance, cf.~Theorem~\ref{thm:custom_crn}.
This leads to the following update rule for $\Nwa$ after the transmission of the $l$-th pilot symbol
\begin{equation}
    \Nwa[l+1] =    \begin{cases}
                        \Nwa[l] & \text{,\;\;\;if}\;\; \xhat=\xtrue \\
                        \Nwa[l]+1 & \text{,\;\;\;if}\;\; \xhat=1 \text{ and } \xtrue = 0 \\
                        \Nwa[l]-1 & \text{,\;\;\;if}\;\; \xhat=0 \text{ and } \xtrue = 1
                    \end{cases}.
                    \label{eq:crn_learning_rule}
\end{equation}
For the practical realization of \eqref{eq:crn_learning_rule}, a reservoir of $\Nwi$ inactive weight molecules $\Woff$ that can be converted into $\Won$ molecules (and vice versa) is introduced into the receiver cell.
At any point in time, $\Nwa+\Nwi=\Nw$, where $\Nw \in \mathbb{N}$ remains constant over time.
Here, $\Nw$ should be chosen large enough such that \ac{MAP} detection can be realized for any possible value of $\threshold$ by $\crn{LC}$.
Since $\threshold \leq \Nr$, one could choose any $\Nw\geq\left\lceil\frac{\rrc{off}}{\rrc{on}}\Nr\right\rceil$.

Learning rule \eqref{eq:crn_learning_rule} is readily implemented by the following chemical reactions
\begin{align}
        \Xhaton + \XtrueOff + \Woff &\xrightarrow{\rrc{u,1}} \Xhaton + \Won,\label{eq:CRN_learning_1}\\
        \Xhatoff + \XtrueOn + \Won &\xrightarrow{\rrc{u,2}} \Xhatoff + \Woff,\label{eq:CRN_learning_2}
\end{align}
where $\rrc{u,1}$ and $\rrc{u,2}$ are the respective reaction rate constants.
In \eqref{eq:CRN_learning_1}, the simultaneous presence of one $\Xhaton$ and one $\XtrueOff$ molecule activates one weight molecule, corresponding to the case $\xhat = 1$, $\xtrue = 0$ in \eqref{eq:crn_learning_rule}.
Analogously, \eqref{eq:CRN_learning_2} implements the case $\xhat = 0$, $\xtrue = 1$ in \eqref{eq:crn_learning_rule}.
Since the $\XtrueOff$ ($\XtrueOn$) molecule is consumed in \eqref{eq:CRN_learning_1} (\eqref{eq:CRN_learning_2}), it is guaranteed that exactly one weight molecule is activated (deactivated) per pilot symbol for sufficiently long update intervals.

%% file: sections/performance_evaluation.tex
\section{Performance Evaluation}
\label{sec:performance_evaluation}

In this section, the performance of the receivers proposed in Sections~\ref{sec:BM_MAP} and \ref{sec:low_complexity_MAP} is evaluated for a specific MC system.
Since the proposed receivers do not rely on a specific channel model, they are readily applicable to other \ac{MC} channels than the one considered here as long as the \ac{MAP} detector coincides with a threshold detector.
The detector proposed in Section~\ref{sec:low_complexity_MAP} can even be applied if no channel model is available at the design time.

\subsection{Evaluation Setup}
We evaluate our proposed detectors using a similar channel model as in \cite{kuscu2018_ml_detection}. Namely, the ligand concentration for a transmitted symbol $x$ around the receiver is given by
\begin{equation}
    \cx = \cn + \deltac \cdot x\label{eq:concentration},
\end{equation}
where $\cn$ and $\deltac$ denote the concentration due to the expected background noise and the concentration increase due to the release of molecules if $x=1$ is transmitted, respectively.

The receiver cell uses $\Nr \in \{30,50\}$ \ac{iid} receptors to estimate the ligand concentration and thus the transmitted symbol $x$. From \cite{kuscu2018_ml_detection}, we obtain the binding probability of receptor $i$ for $x$ as
\begin{equation}
    \Pr[Y_i=1|X=x] = \frac{c_x}{c_x + \frac{k_-}{k_+}}.
\end{equation}
Here, $k_+$ and $k_-$ are the binding and unbinding rate constants between ligand and receptor, respectively.

For equiprobable symbols, the joint distribution $q_{\Z}(\z)=q_{\Y,X}(\y,x)$ is then given by
\begin{equation}
    q_{\Y,X}(\y,x) = \frac{1}{2} \prod_{i=1}^{\Nr} \Pr[Y_i=y_i|X=x].\label{eq:system_model_joint_distribution}
\end{equation}

We choose $k_+=2 \cdot 10^{-19}\meter^3\second^{-1}$ and $k_-=20\second^{-1}$ as in \cite{kuscu2018_ml_detection}. We also consider the system model for diffusive ligand propagation and instantaneous molecule release from \cite{kuscu2018_ml_detection}, given by
\begin{equation}
    c(\tau) = \frac{\gamma}{(4 \pi D \tau)^{3/2}}  \exp\left(-\frac{d^2}{4 D \tau}\right).
\end{equation}
Here, $d$, $D$, $\gamma$ and $\tau$ denote the distance between receiver and transmitter, the diffusion coefficient, the number of released molecules, and the time since the molecule release, respectively. For $\gamma=10^3$, $D=10^{-10}\frac{ \si{\square\metre} }{\si{\second}}$, and $d=0.75\si{\micro\metre}$, $c(\tau)$ has a peak value in the order of $10^{20} \frac{\text{molecules}}{\si{\cubic\metre}}$. Thus, we choose $\deltac=1.5 \cdot 10^{20} \frac{\text{molecules}}{\si{\cubic\metre}}$. For the noise levels, we assume two scenarios. The first one with $\cn=2.5 \cdot 10^{19} \frac{\text{molecules}}{\si{\cubic\metre}}$ and the second one with $\cn=1.0 \cdot 10^{19} \frac{\text{molecules}}{\si{\cubic\metre}}$.

\subsection{Training Boltzmann Machines}
\begin{figure}
    \centering
    \includegraphics[width=0.5\textwidth]{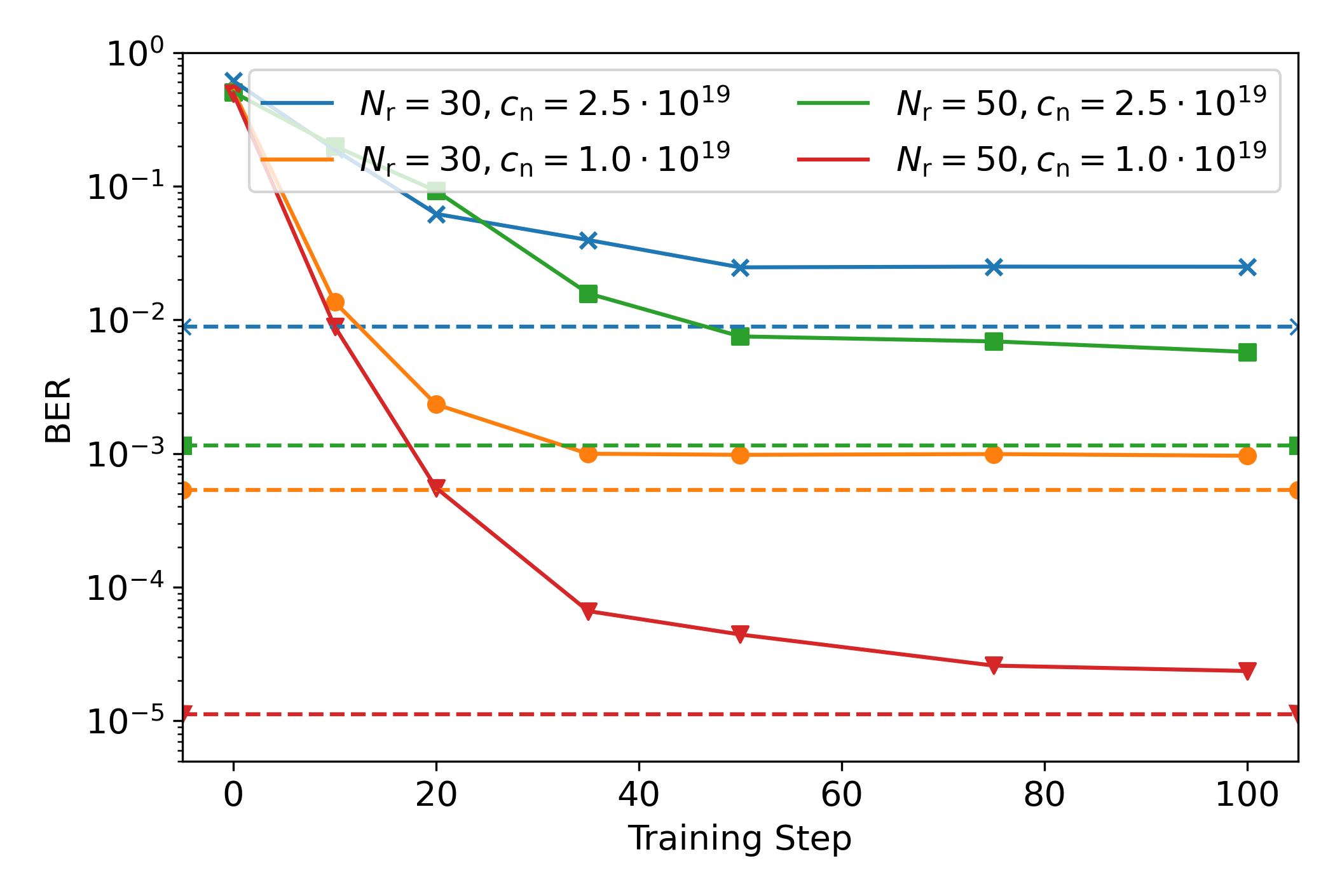}
    \vspace*{-10mm}
    \caption{Average \acp{BER} as a function of the training step for the four considered scenarios (solid lines). After a fast initial decay, the \acp{BER} approach the optimal MAP performance (dashed lines) slowly.}
    \label{fig:BER_time}
    \vspace*{-5mm}
\end{figure}
We train the \acp{BM} with estimated expectations $\expectation{q_{\Y,X}}{\z}$ and $\expectation{q_{\Y,X}}{\z \z^\transpose}$ based on $10^4$ samples from (\ref{eq:system_model_joint_distribution}).
In each training step, the first- and second-order moments of the \ac{BM} for the current biases and weights need to be estimated. To this end, $10^4$ Gibbs samples are generated using the Gibbs sampling algorithm \cite{mackay2003information}. This algorithm exploits that (\ref{eq:bm_conditional_probability}) can be easily evaluated and returns samples distributed according to the Boltzmann distribution $p_{\hat{X},\Y}(\xhat,\y)$. From the obtained Gibbs samples, we then compute the expectations $\expectation{p_{\hat{X},\Y}}{\z}$ and $\expectation{p_{\hat{X},\Y}}{\z \z^\transpose}$.

Similar as in deep learning \cite{Goodfellow-et-al-2016}, an adaptive learning rate $\eta[l]$ is used. In training steps $l\in\{0,\dots,19\}$, $\eta[l]=1.0$, for $l\in\{20,\dots,49\}$, $\eta[l]=0.5$, and for $l\in\{50,\dots,99\}$, $\eta[k]=0.1$.

For the initialization of the weights, first the matrix $\W'=\frac{1}{2}(\V+\V^\transpose) \in \mathbbm{R}^{(\Nr+1) \times (\Nr+1)}$ is defined, where the entries of random matrix $\V$ are \ac{iid} Gaussian \acp{RV} with zero mean and variance $\frac{1}{\Nr+1}$. The initial weight matrix $\W_0$ is obtained from $\W'$ by setting the diagonal entries to zero. To reduce the number of parameters, we also set all entries that capture only correlations among receptors to zero, as well, and we also do not update them during training.

We train 20 \acp{BM} for each scenario using this approach. In Fig.~\ref{fig:BER_time}, the obtained \acp{BER} are shown as a function of the training step. For each \ac{BM}, the \ac{BER} is computed by comparing $\hat{x}$ to $x$ until 100 errors are made. Then, the \acp{BER} of all \acp{BM} are averaged. For reference, the \ac{BER} obtained from the corresponding \ac{MAP} detectors are also shown in Fig.~\ref{fig:BER_time} (dashed lines).

Fig.~\ref{fig:BER_time} shows that the \acp{BER} achieved with \acp{BM} that were sufficiently trained (for more than 50 training steps) decrease as the number of receptors increases and increase as the background noise increases.
Furthermore, Fig.~\ref{fig:BER_time} confirms that the considered \acp{BM} approach close-to-\ac{MAP} performance as they are trained over more and more training steps.
The remaining gap between the \acp{BM}' performance and the \ac{MAP} detector visible in Fig.~\ref{fig:BER_time} results from the limited amount of data that is utilized to train the \acp{BM}.
With more data samples and longer training, the \ac{BM} would come even closer to \ac{MAP} performance.

\subsection{Convergence of CRN Online Learning}
\label{sec:performance_evaluation_crn_convergence}
\begin{figure}
    \centering
    \includegraphics[width=0.5\textwidth]{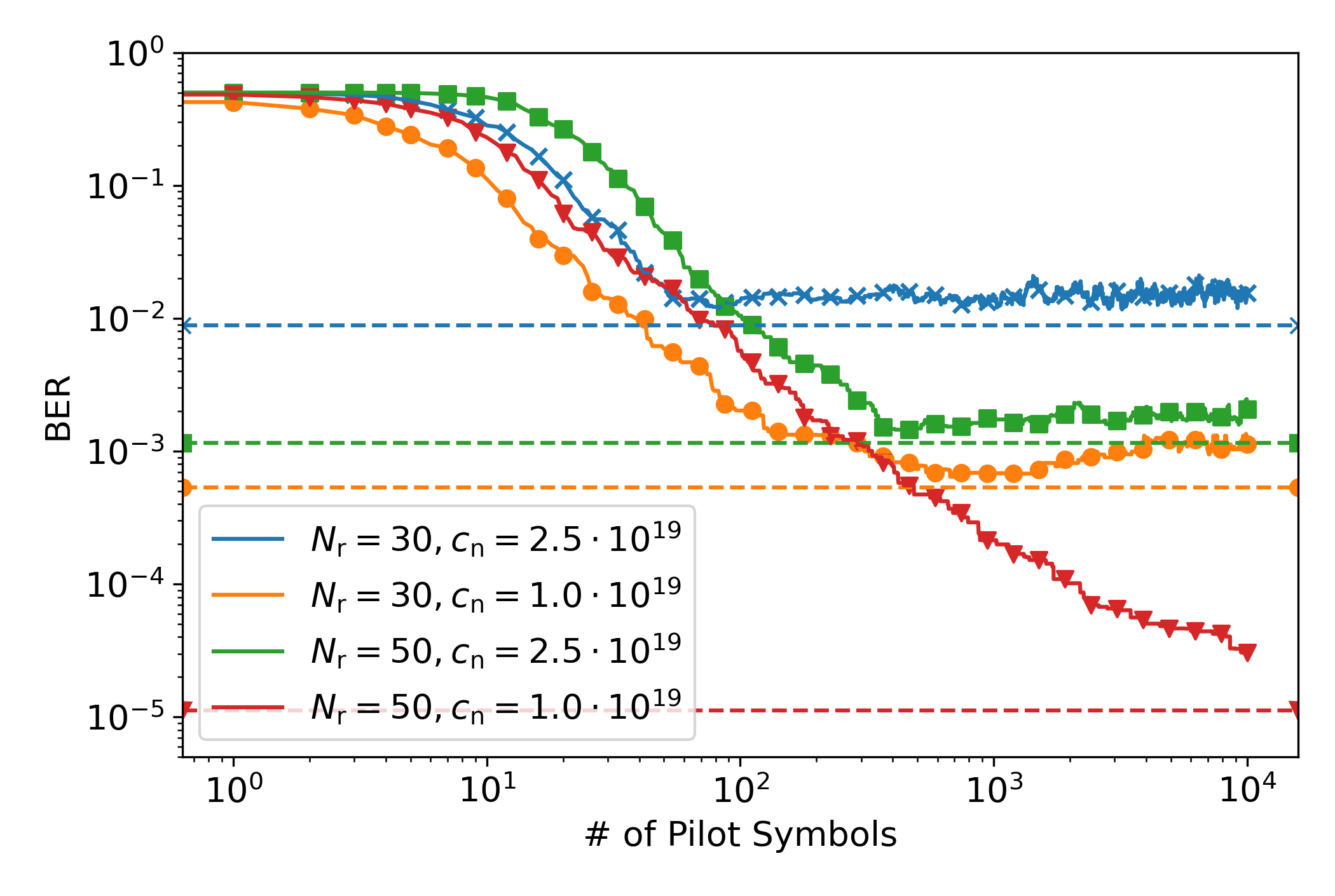}
    \vspace*{-10mm}
    \caption{Average \acp{BER} as obtained from $\crn{LC}$ as functions of the number of pilot symbols for the four considered scenarios (solid lines). All systems eventually approach MAP performance (dashed lines).}
    \label{fig:BER_online_learning}
    \vspace*{-5mm}
\end{figure}
Next, we investigate the performance of the adaptive \ac{CRN} proposed in Section~\ref{sec:low_complexity_MAP}. To this end, we employ (\ref{eq:crn_learning_rule}) with $\Nwa[0]=0$ as initial value.
After each pilot symbol, the \ac{BER} for the current value of $\Nwa$ is computed. This procedure is repeated 20 times for each scenario and the obtained \acp{BER} are averaged. In this way, it can be observed how many pilot symbols are necessary to achieve a certain performance. The corresponding \acp{BER} are shown in Fig.~\ref{fig:BER_online_learning}. 

We observe from Fig.~\ref{fig:BER_online_learning} that the \ac{CRN}-based \ac{MAP} detector proposed in Section~\ref{sec:low_complexity_MAP} is indeed able to learn the optimum decision threshold $\nu$ from the pilot symbols.
This is a very powerful feature, since the receiver can be deployed in \ac{MC} systems even if prior knowledge about the specific channel is not available.
Moreover, since the decision threshold $\nu$ is not hard-wired in the implementation of the receiver, it can even adapt if $\nu$ is time-varying.

However, learning rule \eqref{eq:crn_learning_rule} also introduces some suboptimality because it depends on individual observations of the receptor occupancy in the pilot symbols intervals. Because these are subject to noise, e.g., due to the random interactions between ligands and receptors, \eqref{eq:crn_learning_rule} can in some cases cause updates of $\Nwa$ that slightly decrease performance, like for the orange line in Fig.~\ref{fig:BER_online_learning} after $\approx10^3$ pilot symbols.

This performance degradation is an inherent property of the learning rule \eqref{eq:crn_learning_rule} because it is based on individual transmitted symbols; even if $\Nwa$ has the correct value for \ac{MAP} detection, it still can happen that $\xhat \neq \xtrue$ simply due to the intrinsic randomness of the channel or the receptor occupation. Then, $\Nwa$ is updated wrongly which causes suboptimal performance. In contrast to \ac{BM} learning, simply using more pilot symbols cannot remedy this, as long as each individual update of $\Nwa$ still relies on individual symbols. 

\subsection{Time-Variant Background Noise Levels}
\begin{figure}
    \centering
    \includegraphics[width=0.5\textwidth]{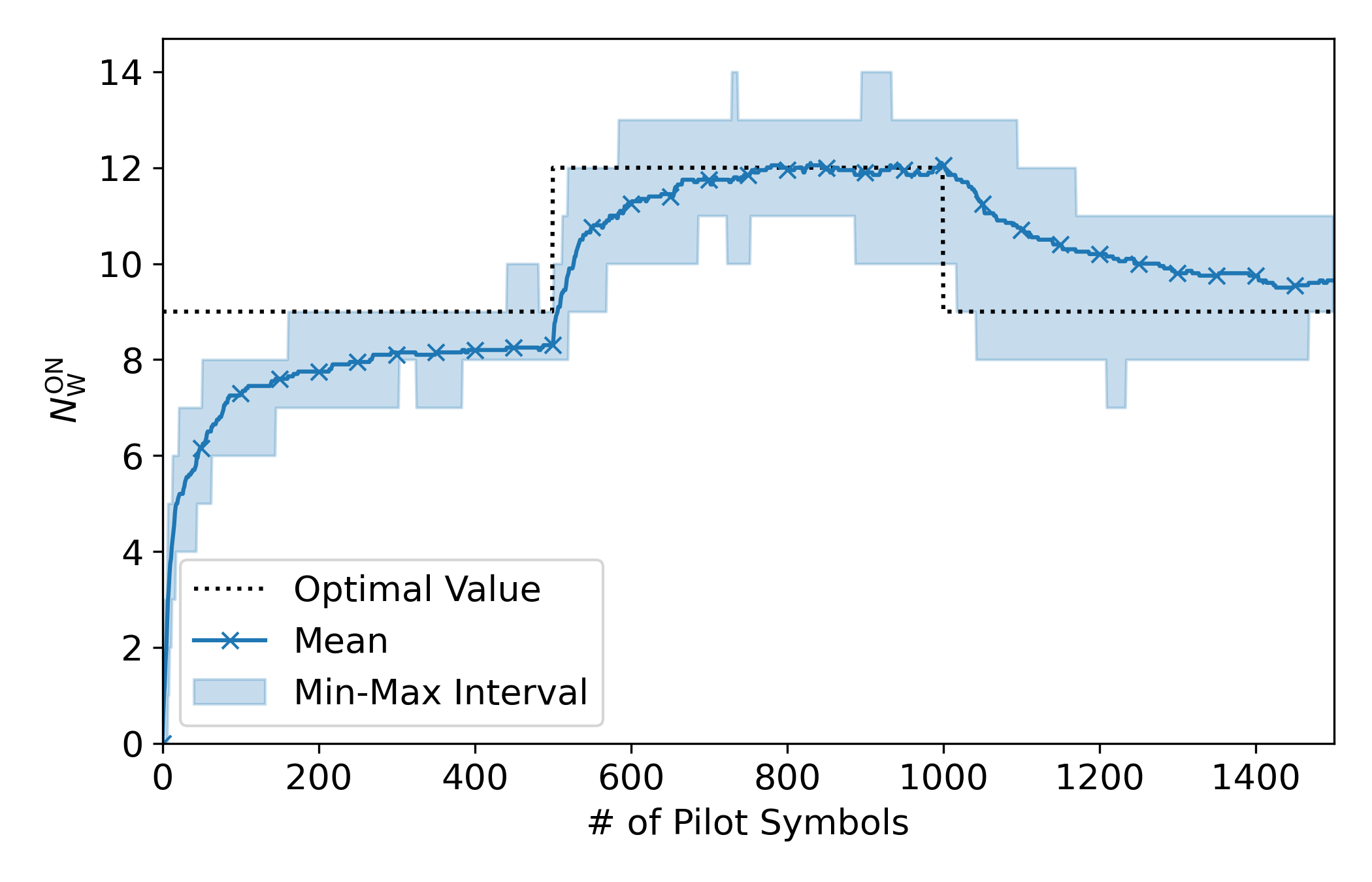}
    \vspace*{-10mm}
    \caption{Average number of $\Nwa$ over time (solid line) and interval with minimum and maximum value of $\Nwa$ (shaded area) for time-variant background noise levels. For reference, the optimal values are shown by the dotted black line.}
    \label{fig:BER_time_changing_csi}
    \vspace*{-5mm}
\end{figure}
In real-world applications, channel parameters can be time-variant, e.g., background noise levels could depend on the activity of other users in the channel. To explore the impact of this on the performance of the online-learning rule proposed in Section~\ref{sec:low_complexity_MAP}, we consider a receiver with $\Nr=30$ receptors and assume $\deltac=(1.5 \cdot 10^{20}-1.0 \cdot 10^{19}) \frac{\text{molecules}}{\meter^3}$. Initially, we set $\cn_1=1.0 \cdot 10^{19} \frac{\text{molecules}}{\meter^3}$ as noise level, change it after 500 pilot symbols to $\cn_2=2.5 \cdot 10^{19} \frac{\text{molecules}}{\meter^3}$ and after 500 more pilot symbols back to the original value.

The resulting average number of active weight molecules, representative of the detection threshold of the \ac{CRN}-based detector, is shown in Fig.~\ref{fig:BER_time_changing_csi} over time. The shaded area shows the lowest and highest observed values of $\Nwa$ for the considered 20 model realizations. We observe from Fig.~\ref{fig:BER_time_changing_csi} that $\Nwa$ follows the optimal values relatively closely over time. Still, the current learning rule requires a considerable amount of pilot symbols which indicates that the proposed detector is best suited for slowly changing channels and relatively high \acp{BER}.

%% file: sections/conclusion.tex
\section{Conclusion}
\label{sec:conclusion}
In this paper, we have introduced two \acp{CRN} that can be used to realize \ac{MAP} detection for appropriately chosen reaction rate constants and molecule counts. In contrast to existing approaches, both detectors exploit the intrinsic stochastic fluctuations of chemical reactions.

The first detector exploits that \acp{BM} can be represented using \acp{CRN}. It can be trained offline using either an analytical channel model or simulations. We demonstrated for an example system model that even with relatively few training data close-to-\ac{MAP} performance can be achieved.
The second detector is based on a custom designed \ac{CRN} that can be trained with pilot symbols even after deployment in a cellular receiver. We showed that this detector achieves not only close-to-\ac{MAP} performance, but that it can even adapt to changes in the channel.
The proposed detectors present a step towards practically realizable and adaptive receivers for future nano-scale \ac{MC} systems.

In future research, the proposed adaptive \ac{CRN}-based detector could be improved by designing a learning rule that is more robust towards signal noise or has a larger step size for faster changing channels. Also, it would be interesting to generalize the proposed detectors to M-ary transmission schemes.